%% file: isit20_ncw.tex
\documentclass[conference,letterpaper]{IEEEtran}

\usepackage{geometry}
\usepackage[utf8]{inputenc} 
\usepackage[T1]{fontenc}
\usepackage{url}
\usepackage{ifthen}
\usepackage{cite}
\usepackage[cmex10]{amsmath} % Use the [cmex10] option to ensure complicance                           % with IEEE Xplore (see bare_conf.tex)
\usepackage{amssymb}
\usepackage{graphicx}
\usepackage{bm}
\usepackage{bbm}
\usepackage{color}
\usepackage{subfig}
\usepackage{mathtools}
\usepackage{algorithm}
\usepackage{algorithmic}
\usepackage{graphicx}
\usepackage{float}
\input{./theoremstylesheet.tex}

\input{./macro_ihwang.tex}

\input{./macro_comove.tex}

\begin{document}
%\title{On the Noisy Coin Weighing Problem} 
%\title{On the Noisy Combinatorial Quantitative Group Testing Problem}
%\title{On the Combinatorial Quantitative Group Testing Problem with Noisy Query and Partial Recovery}
\title{Positional Identifiability from Pairwise Collision Data}

\author{
\IEEEauthorblockN{Yun-Han Li, Ilan Shomorony and Olgica Milenkovic
}
\IEEEauthorblockA{
%\IEEEauthorrefmark{1}
\\Department of Electrical and Computer Engineering, University of Illinois Urbana-Champaign, USA\\ 
Email: \url{yunhanl2,ilans,milenkov@illinois.edu}
}
\IEEEauthorblockA{
%\IEEEauthorrefmark{1}
}
}
\maketitle

\begin{abstract}
We study the problem of recovering the relative positions of objects moving along the real line based only on pairwise collision data. While interaction-based sensing systems arise naturally 
in a variety of practical settings, a systematic theoretical 
understanding of positional identifiability from collision 
observations alone remains unexplored. Our contributions are three-fold. First, under the full 
observability model, in which both the set of collisions and their 
temporal ordering are known, we show that the relative 
positions of all objects can be uniquely recovered if and only 
if the collision history, represented as a graph, is connected. Second, we show that under partial 
observability, where only colliding pairs are observed without timing information, the problem is related to \emph{function graphs} and introduce a canonical 
layer decomposition in which each layer corresponds to a 
maximal clique; the contraction graph induced by this 
decomposition is an interval graph, and we provide efficient 
algorithms to recover it. Third, under incomplete observations where even some pairwise collision observations may be missing, we formulate the problem as a graph completion problem and establish its NP-hardness via a $4$-approximation relationship with the graph bandwidth problem.
\end{abstract}

\section{Introduction}\label{sec:intro}

\input{\srcpath introduction.tex}

%\section{Review of Deduplication Approaches}\label{sec:practice}
%\input{\srcpath sec_practice.tex}

\section{Problem Formulation}\label{sec:formulation}

\input{\srcpath sec_formulation.tex}

\section{Results}\label{sec:result}
\input{\srcpath sec_result.tex}

%\section{The Jump Metric: Results}\label{sec:jump}
%\input{\srcpath sec_jump.tex}

%\section{Related Works}\label{sec:related}
%\input{\srcpath sec_related.tex}

%\section{Future Work}\label{sec:future}
%\input{\srcpath sec_future.tex}

%\section{Appendix}\label{sec:appendix}
%\input{\srcpath appendix.tex}

%\section{Main Results}\label{sec:results}
%\input{\srcpath sec_results.tex}

%\section{Algorithms}\label{sec:algo}
%\input{\srcpath sec_results.tex}
%\input{\srcpath sec_algo.tex}

%\section{Concluding Remarks}\label{sec:conclusion}
%\input{\srcpath sec_conclusion.tex}

%\section*{Acknowledgments}
%\input{\srcpath acknowledgment.tex}

%This work was supported by the National Science Foundation (NSF) under grants CCF 2107344 and CCF 2046991.

\bibliographystyle{IEEEtran}
%\vfill
\bibliography{Ref.bib}

%\newpage
%\onecolumn

%\appendix
%\input{\srcpath sec_proofs.tex}

%%%%%%
%% To balance the columns at the last page of the paper use this
%% command:
%%
%\enlargethispage{-1.2cm} 
%%
%% If the balancing should occur in the middle of the references, use
%% the following trigger:
%%
%\IEEEtriggeratref{3}
%%
%% which triggers a \newpage (i.e., new column) just before the given
%% reference number. Note that you need to adapt this if you modify
%% the paper.  The "triggered" command can be changed if desired:
%%
%\IEEEtriggercmd{\enlargethispage{-20cm}}
%%
%%%%%%

%%%%%%
%% References:
%% We recommend the usage of BibTeX:
%%

\end{document}

%% file: theoremstylesheet.tex
\newtheorem{example}{Example}[section]
\newtheorem{lemma}{Lemma}[section]
\newtheorem{theorem}{Theorem}[section]
\newtheorem{proposition}{Proposition}[section]
\newtheorem{corollary}{Corollary}[section]
\newtheorem{definition}{Definition}[section]

%% file: macro_ihwang.tex
\newcommand{\lp}{\left(}
\newcommand{\rp}{\right)}

\newcommand{\argmax}{\mathop{\mathrm{arg\,max}}}

%% file: macro_comove.tex
\newcommand{\Max}[1]
{#1_{\max}}
\newcommand{\Min}[1]
{#1_{\min}}
\newcommand{\leqs}{\leq}
\newcommand{\geqs}{\geq}
\newcommand{\br}[3]{#1-#2\sim#3}
\newcommand{\bn}[1]{N^{+}\lp#1\rp}
\newcommand{\W}{W}

\newcommand{\mv}{v}
\newcommand{\mV}{V}

\newcommand{\e}{e}
\newcommand{\E}{E}
\newcommand{\h}{h}
\newcommand{\hg}{H}

\newcommand{\p}{p}

\newcommand{\pe}[1]{p\lp #1\rp}

%% file: introduction.tex
In many practical distributed sensing systems, global knowledge must be inferred
based on observed local interactions~\cite{eagle2006reality, lazer2009computational}.
Most often, the full history of the system evolution is not available, and one
only has access to information pertaining to ``local events'' such as object
collisions or interactions. A natural question in this setting is whether such
local interaction information allows one to recover an underlying position of
objects or the geometry of object movements and encounters.

Here, we consider a combinatorial model in which a collection of objects move
along the real line (or an interval) according to unknown continuous trajectories. Instead of
observing the exact positions of the objects, we only observe pairwise collisions on the line, with or without timestamps for the events. The question of interest is to determine to what extent can one reconstruct
the relative positions of objects from such collision information.

Interaction-based observations arise naturally in a variety of settings, including particle physics or sensing
systems with limited resolution~\cite{eagle2006reality}. In these scenarios,
precise positional information of the objects is not directly observable, and one
only has access to information about which pairs of objects interacted. As a result,
the position of the objects must be inferred indirectly.

The goal of this paper is to demonstrate that the problem can be cast in a combinatorial framework in which
objects are represented by continuous functions, and observations
correspond to pairwise intersections of the functions. The resulting formulation 
is closely related to \emph{function graphs}, a class of intersection graphs in
which vertices correspond to one-dimensional continuous functions and edges
indicate the existence of an intersection~\cite{ehrlich1976intersection,
scheinerman1984intersection}. More broadly, this places our problem within the
extensive literature on intersection graphs, including interval
graphs~\cite{fulkerson1965incidence, booth1976testing}, circular-arc graphs,
permutation graphs~\cite{pnueli1971transitive}, and string
graphs~\cite{kratochvil1994intersection, golumbic2004algorithmic, mckee1999topics}.
While these graph classes have been extensively studied from the perspectives of
recognition and structural characterization~\cite{golumbic2004algorithmic,
mckee1999topics}, the question of what positional information can be recovered from intersection data is new. In
particular, we investigate how different levels of observability, such as temporal information or the presence of missing interactions,
affect our ability to recover the relative positions of objects.

Formally, we consider $n$ objects $V = \{v_1, v_2, \ldots, v_n\}$ whose positions evolve 
according to unknown continuous trajectories $v_i(t) : \mathbb{R} \to \mathbb{R}$. A 
\textbf{collision} between objects $i$ and $j$ occurs at time $t$ if $v_i(t) = v_j(t)$. 
We study three levels of observability and the resulting types of collision data:

\begin{enumerate}
    \item \textbf{Full observability} (collision times are known): The observation is the 
    complete \emph{ordered (wrt time)} sequence of collisions $h = e_1, e_2, \ldots$, equivalently 
    represented as an edge-labeled multigraph $H = (V, E)$ where edge labels encode 
    temporal order. We show that the relative positions $p_k$ of all objects at all time can be uniquely recovered (up to reversal) if and only if $H$ is connected.

    \item \textbf{Partial observability} (no timing information): The collisions are described by a graph $G = (V, E)$, where $uv \in E$ indicates that $u$ 
    and $v$ had collided at \textit{some} time, without ordering information. This model 
    is closely related to \textit{function graphs}~\cite{ehrlich1976intersection,
    golumbic2004algorithmic,mckee1999topics}, a well-studied class of intersection 
    graphs. Exact recovery of positions is generally impossible in this setting. 
    Instead, we introduce a canonical layer decomposition 
    $(V_1, V_2, \ldots, V_k)$ -- a partition of $V$ where each layer $V_i$ is a 
    maximal clique of objects whose relative order cannot be distinguished from unordered 
    collisions alone. We show that the contraction graph induced by this decomposition 
    is an interval graph~\cite{fulkerson1965incidence,booth1976testing} and 
    provide efficient algorithms to recover it.

    \item \textbf{Incomplete observations} (missing collisions): The observation is a 
    subgraph $G = (V, E)$ of the true collision graph, where some collision edges may 
    be absent. We study the problem of completing $G$ into a valid function graph 
    $G' = (V, E')$ with $E \subseteq E'$ while preserving as much positional information 
    as possible. We show that a natural formulation of this completion problem is a $4$-approximation for the graph bandwidth problem~\cite{diaz2002survey,
    petit2013addenda}, implying NP-hardness~\cite{feige2009approximating} and connecting our setting to classical graph layout problems~\cite{diaz2002survey,petit2013addenda,
    dunagan2001euclidean,feige2009approximating}.
\end{enumerate}
Our work takes a step towards a systematic understanding of positional 
identifiability from pairwise interaction data, bridging ideas from intersection graph~\cite{ehrlich1976intersection,golumbic2004algorithmic,
mckee1999topics} and graph layout problems~\cite{diaz2002survey,petit2013addenda}.

%% file: sec_formulation.tex
We start by introducing our model movement constraints and observation assumptions.

\underline{Movement model.} There are $n$ objects, $\mV=\{\mv_{1},\mv_{2},\hdots,\mv_{n}\},$ whose positions change in time following trajectories described by continuous functions $\mv_{i}(t):\mathbb{R}\rightarrow \mathbb{R}$. The trajectories are not known \emph{a priori}. 

\underline{Observation model.} We say that a \emph{collision} $\e=\mv_{i}\mv_{j}$ of two objects $i$ and $j$, or correspondingly, their trajectories $\mv_{i},\mv_{j}$, occurs  at time $t$ if $\mv_{i}(t)=\mv_{j}(t)$ (i.e., the trajectories intersect). For simplicity, we assume that (1) collisions are instantaneous so that the trajectories intersect only at one time point $t$ instead on overlapping on a time interval $[t^{-},t^{+}]$; This is possible since we do not put any constraints on hidden trajectories except continuity.
%\textcolor{red}{I am not sure that this is possible unless the two objects move in opposite directions through $t$; as your model involves movement on a line, the functions you will have are going to be piecewise linear, right? because you can only change direction; something is not adding up here?}
%\textcolor{blue}{They might not be linear. For example, two cars driving on a highway toward the same direction doesn't imply they will always next to each other. Sometimes one is faster than another. Their position vs time might not be linear } %\textcolor{red}{disagree, unless you assume that drivers can change velocity continuously; at least say that we assume that the movement on the line leads to continuous position functions and mention piecewise liner later on when you talk about graph, and cite papers of prof. West. Also, at some point you want to mention community detection and co-movement communities, since the problem originated from this investigation - your communities are the "extremal movers"}
(2) there is at most one collision among trajectories for any $t$, which also implies that no three or more trajectories intersecting at the same time point;
(3) $v_i(t)$ and $v_j(t)$ always cross when a collision happens
(i.e., $\mv_{i}(t^{-})<\mv_{j}(t^{-})\Rightarrow \mv_{i}(t^{+})>\mv_{j}(t^{+})$). The \emph{history of collisions} is a sequence of ordered collision times $\h=\e_{1},\e_{2},\hdots$, where $\e_{i}$ is the $i$-th collision. Equivalently, the collision history can be characterized by the edge-labeled collision multigraph $\hg = (\mV,\E=\{\e_{i}\})$, where the edge labels indicate the collision observation orders.
\begin{example}
    The collision graph corresponding to the functions in Figure~\ref{shrink} (left) is $G=(\mV,\E)$ with $\mV=\{\mv_{i}\}_{i=1}^{4}$ and $\E=\{e_{1}=\mv_{3}\mv_{4},e_{2}=\mv_{2}\mv_{3},e_{3}=\mv_{1}\mv_{4},e_{4}=\mv_{1}\mv_{4},e_{5}=\mv_{2}\mv_{4},e_{6}=\mv_{3}\mv_{4}\}$.
\end{example}
\underline{Goal.} We say that $\p_{k}=\lp\mv_{i_{1}},\mv_{i_{2}},\mv_{i_{3}},\hdots\rp$ are the positions   
% \iscomment{ordering?} \textcolor{red}{If we use ordering, people might confuse it with ordering of collisions?} 
of vertices $\mV=\{\mv_{1},\mv_{2},\hdots,\mv_{n}\}$ between consecutive collisions $\e_{k}$ and $\e_{k+1}$ if $\mv_{i_{1}}(t)<\mv_{i_{2}}(t)<\mv_{i_{3}}(t)<\hdots$ for all $t$ between the $k$-th and $(k+1)$-th collision time. 
Moreover, the end-position $\pe{\hg}$ represents the underlying position after the last collision. Given a complete (partial) collision history, we aim to recover $\p_k$ 
for all $k$. For instance, in Figure~\ref{shrink}, $\p_0 = (\mv_1, \mv_3, \mv_5, \mv_2, \mv_4)$ 
and $\p_1 = (\mv_1, \mv_5, \mv_3, \mv_2, \mv_4)$, where $\p_1$ is obtained from 
$\p_0$ by swapping $\mv_3$ and $\mv_5$ after collision $e_1$.
Note that replacing $\{\mv_{i}(t)\}$ with $\{-\mv_{i}(t)\}$ would result in the same history, but with all $\p_{k}$s in inverse order. Therefore, recovered positions are only unique up to inversion.

In Section~\ref{sec_ch}, 
we consider the fully observable setting in which both the set 
of collisions and their temporal ordering are known, and 
characterize the exact condition under which all object positions 
can be uniquely recovered. In Section~\ref{wotwoel}, we relax 
this assumption and study the case where only the set of colliding 
pairs is observed without timing information; we show that a 
canonical layer decomposition can still be recovered. Finally, in 
Section~\ref{sec:missing}, we consider the most challenging 
setting where the observed collision graph may be incomplete, and 
establish a connection to the graph bandwidth problem that implies 
NP-hardness of the completion problem.

%% file: sec_result.tex
%We now present our main results, organized by the level of observability of the collision data.

\subsection{Exact Recovery Under Full Observability}\label{sec_ch}
We start by establishing the exact condition for recovery of all $\p_{k}$ given a complete collision history.
\begin{lemma}\label{ch_uep}
There is a unique $\pe{\hg}$ iff $\hg$ is connected.
\end{lemma}
\begin{IEEEproof}
We first show that $\hg$ being connected implies a unique $\pe{\hg}$. We prove this by induction on the number of collisions. Let $\hg_{\leq k}$ be the partial history that keeps the first $k$ collisions of $\hg$. Let $\e_{k}=\mv_{i}\mv_{j}$ be the last collision in $\hg_{\leq k}$.
For $k=1$, clearly $\pe{\hg_{\leq 1}}= \lp\mv_{i},\mv_{j}\rp$ since there is only one collision between two vertices.
For $k>1$,
\begin{enumerate}
    \item If $\hg_{\leq k-1}$ is connected, then by assumption, $\pe{\hg_{\leq k-1}}$ is unique. By swapping $\mv_{i},\mv_{j}$ in $\pe{\hg_{\leq k-1}}$, we have the unique $\pe{\hg_{\leq k}}$.
    \item If $\hg_{\leq k-1}$ is disconnected, let $U,W$ be two connected components in $\hg_{\leq k-1}$ with $\mv_{i}\in U, \mv_{j}\in W$. 
    By assumption, we have a unique end-position $\lp u_{1},\hdots,u_{|U|}\rp$ for $U$, and $\lp w_{1},\hdots,w_{|W|}\rp$ for $W$.
    Clearly, $\mv_{i}\in\{u_{1},u_{|U|}\},\mv_{j}\in\{w_{1},w_{|W|}\}$, say $\mv_{i}=u_{1},\mv_{j}=w_{1}$.
    Then, by first flipping $\lp u_{1},\hdots,u_{|U|}\rp$ and concatenating it with $\lp w_{1},\hdots,w_{|W|}\rp$, and then swapping $u_{1},w_{1}$, we obtain the unique $\pe{\hg_{\leq k}}=\lp u_{|U|},\hdots,w_{1},u_{1},\hdots,w_{|W|}\rp$. The other three cases can be handled similarly. 
\end{enumerate}
If $\hg$ is disconnected we can recover unique end-positions within each component. However, there is no unique overall end-position since there are no collision between the components and one can hence not recover their relative positions. 
\end{IEEEproof}
Lemma~\ref{ch_uep} states that, given a connected $\hg$, one can uniquely determine the final ordering $\p(\hg)$ (up to reversal). 
% And it is not hard to extend to all $\p_{k}$. 
This argument is extended to all $\p_k$ in the following result.
\begin{theorem}\label{ftol}
    There is a unique $\p_{k}, \forall k,$ iff $\hg$ is connected.
\end{theorem}
\begin{IEEEproof}
    By Lemma~\ref{ch_uep}, we can recover the unique end-position $\pe{\hg}$ if $\hg$ is connected. 
    Then, by considering the collision history in reverse order, we can undo each collision, and hence obtain $p_k$ for all $k$. 
    % Then we reverse the history backwards, starting from last collision and then the one before, $\hdots$, etc to get the unique $\p_{k}$ for all $k$. 
    In each step, we undo $\e_{k}=\mv_{i}\mv_{j}$ 
    by swapping $\mv_{i},\mv_{j}$. 
\end{IEEEproof}

\pdfinclusioncopyfonts=1
%\vspace{+0.05in}

\begin{figure}[H]
% \vspace{-5mm}
    \centering
    \includegraphics[scale=0.24]{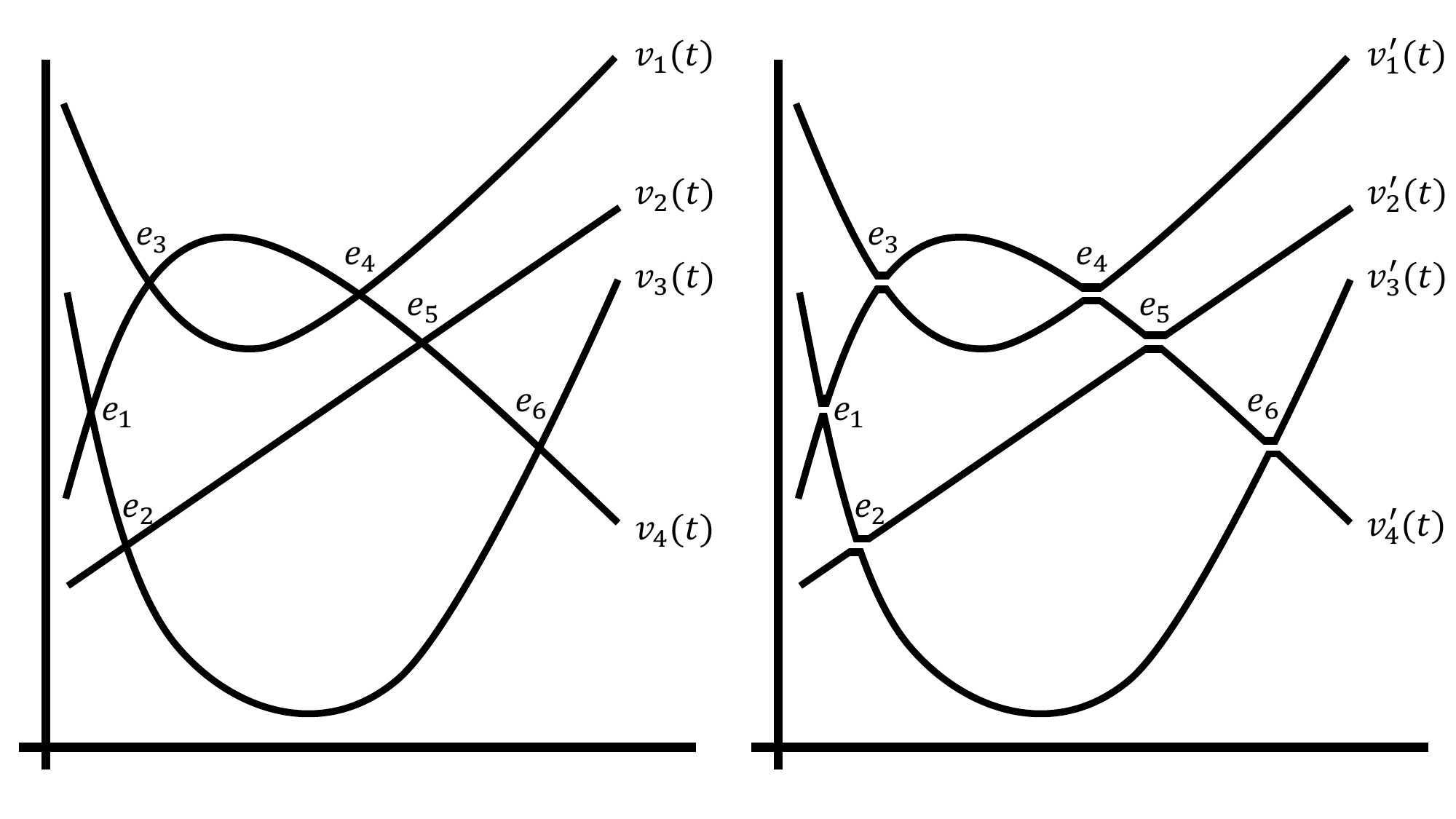}
    \caption{An example with four curves and six collisions. By sequentially swapping $\{\e_{i}\}_{i=1}^{6}$, we make four new curves $\{\mv_{i}'(t)\}_{i=1}^{4}$ without any crossing, and therefore the resulting collision graph is a path.}\label{swapping}
\end{figure}
We provide an alternative proof of Theorem~\ref{ftol} based on sequentially swapping collisions. Formally, swapping a collision $e_k = uv$ means exchanging the future trajectories of $u(t)$ and $v(t)$ after $e_k$ while keeping their past trajectories unchanged, so that the two new trajectories $u'(t)$ and $v'(t)$ touch but do not cross at $e_k$. Equivalently, this operation exchanges all edges incident to $u$ and $v$ that occur after $e_k$ in the collision graph. Applying this swap sequentially to all collisions $e_1, e_2, \ldots$ in order produces a new set of trajectories in which no two curves cross each other and each object collides only with its immediate left and right neighbors, so the resulting collision graph is a path, as illustrated in Figure~\ref{swapping}. Since the positions of objects along a path are immediately readable, all $p_k$ for the new set of trajectories can be recovered directly. Recovering the $p_k$ for the original trajectories then follows by undoing the swaps in reverse order.
%An alternative perspective on Theorem~\ref{ftol} proceeds by sequentially ``turning off'' all collisions. Formally, turning off a collision $e_k = uv$ means swapping the future trajectories of $u(t)$ and $v(t)$ after $e_k$ while keeping their past trajectories unchanged, so that the two new trajectories $u'(t)$ and $v'(t)$ touch but do not cross at $e_k$. Equivalently, this operation exchanges all edges incident to $u$ and $v$ that occur after $e_k$ in the collision graph. Applying this procedure to all collisions sequentially transforms the collision graph into a path, where every object bounces only between its left and right neighbors, as illustrated in Figure~\ref{swapping}. Before turning off collisions, objects move around and collide with many others, making it hard to track their relative positions. After turning off all collisions, however, each object barely moves and only bounces between its left and right neighbors, so their positions become easy to recover directly from the path. Recovering all $p_k$ for the original trajectories then follows by turning all collisions back on in reverse order.
%\iscomment{I'm confused by the paragraph above. What do you mean by an alternative ``perspective''? An alternative proof? Also, calling it ``turning off a collision'' doesn't seem accurate, because (if I understand it correctly), the collision still happens. Should you call ``swapping'' each collision? Also, you should more clearly state that, after swapping each collision sequentially, the resulting graph becomes a line, and briefly explain why. }

\subsection{Layer Decomposition Under Partial Observability}\label{wotwoel}
As shown in Section~\ref{sec_ch}, with a complete history, one can recover all $\p_{k}$ iff $\hg$ is connected. However, in 
practice, the full temporal ordering of collisions may not be available, as is the case in the sensing systems with limited resolution,  
discussed in Section~I. Therefore, we address next the setting without time information. Specifically, our observations form a 
simple graph $G = (\mV, \E)$ with unlabeled edges, where $uv \in \E$ 
indicates that objects $u$ and $v$ collided at some point in time.
This model is closely related to \textit{function graphs}. 
Formally, a graph $G = (V, E)$ is called a function graph if 
there exist continuous functions $f_1, f_2, \ldots, f_n : 
[0,1] \to \mathbb{R}$ such that each vertex $v_i \in V$ 
corresponds to $f_i$, and $v_iv_j \in E$ if and only if 
$f_i(t) = f_j(t)$ for some $t \in [0,1]$. It is known that 
$G$ is a function graph if and only if its complement is a 
comparability graph, which also implies that function graphs 
are recognizable in polynomial time~\cite{golumbic1983comparability}. 
Function graphs generalize permutation graphs, which correspond to the special case where all functions are linear~\cite{golumbic1983comparability}.

In this model, exact recovery of positions is generally 
impossible. Instead, our goal is as follows. We say that $\mv_{i}>\mv_{j}$ iff $\mv_{i}(t)>\mv_{j}(t)$ for all $t$. For any $S\subset\mV$, the subset $\Max{S}:=\{\mv_{i}\in S|\nexists\mv_{j}\in S\text{ with }\mv_{j}>\mv_{i}\}$ is called the max of $S$. 
$\Min{S}$ is defined similarly.
Define $\mV_{1}:=\Max{\mV}$ as the first layer, and $\mV_{i}:=\Max{(\mV-\cup_{j=1}^{i-1}\mV_{j})}$ as the $i$-th layer of $\mV$. 
A layer is a maximal set of objects whose relative order cannot be distinguished based on unordered collisions. In particular, as shown in Corollary~\ref{laycli}, all objects within the same layer must collide with one another.
Given $\hg = (\mV,\E)$, our aim is to recover all layers $\lp\mV_{1},\mV_{2},\hdots,\mV_{k}\rp$ that partition $\mV$.

The notion of a \textit{module} plays a central role in the 
structural analysis of function graphs. A module is a set of vertices that behaves 
``uniformly'' with respect to the rest of the graph -- every 
vertex outside the module is either adjacent to \emph{all}  vertices 
in the module or to \emph{none.} This concept was first introduced 
by Gallai~\cite{gallai1967transitiv} to study the structure 
of comparability graphs, and has since become a fundamental 
tool in the recognition and decomposition of many graph 
classes~\cite{golumbic2004algorithmic}.

\begin{definition}
A set $S \subseteq \mV$ is called a \textit{module} if $|S| > 1$ 
and each $u \in \overline{S}$ connects to all or none of the vertices 
in $S$.
\end{definition}

\begin{proposition}
Let $S = \{\mv_1, \mv_2, \ldots\}$ be a module with 
$T = \{\mv_1(t), \mv_2(t), \ldots\}$ being the set of corresponding trajectories. Define $T'_{\mv_k} := \{\mv'_i(t) := \epsilon \mv_i(t) + \mv_k(t)\}$ as a set of trajectories for $S$ that shrink old trajectories $T$ toward $\mv_k$.
%\iscomment{I find this definition of $T'$ a bit confusing, as it seems that $v_1$ plays a special role? All trajectories in $T$ are being shrunk towards $v_1(t)$? If so, should this be $T'_{v_1}$ or something like that? Also, does $v_1(t)$ become $(1+\epsilon)v_1(t)$? The proof, doesn't seem to make it clear which path the other paths are being shrunk towards}
Then replacing $T$ with $T'_{\mv_k}$
would result in the same function graph.
\end{proposition}
\begin{IEEEproof}
    For $\mv_{i},\mv_{j}\in S$, it is clear that $\mv_{i}(t),\mv_{j}(t)$ intersect iff $\mv'_{i}(t),\mv'_{j}(t)$ intersect. For $\mv_{i}\in S,\mv_{j}\in\overline{S},$ it is also clear that $\mv_{i}(t),\mv_{j}(t)$ intersect iff $\mv_{i}' (t),\mv_{j}(t)$ intersect. Since for $\epsilon \to 0$ all $v'_i(t)$ become arbitrarily close to $v_k(t)$, 
any $v_j(t) \notin S$ that intersects $v_k(t)$ must intersect all $v'_i(t)$, 
and any $v_j(t)$ that does not intersect $v_k(t)$ intersects none of them. See Figure~\ref{shrink} for an illustration.
\end{IEEEproof}
\begin{figure}[t]
 
    \centering
    \includegraphics[width=0.4\textwidth]{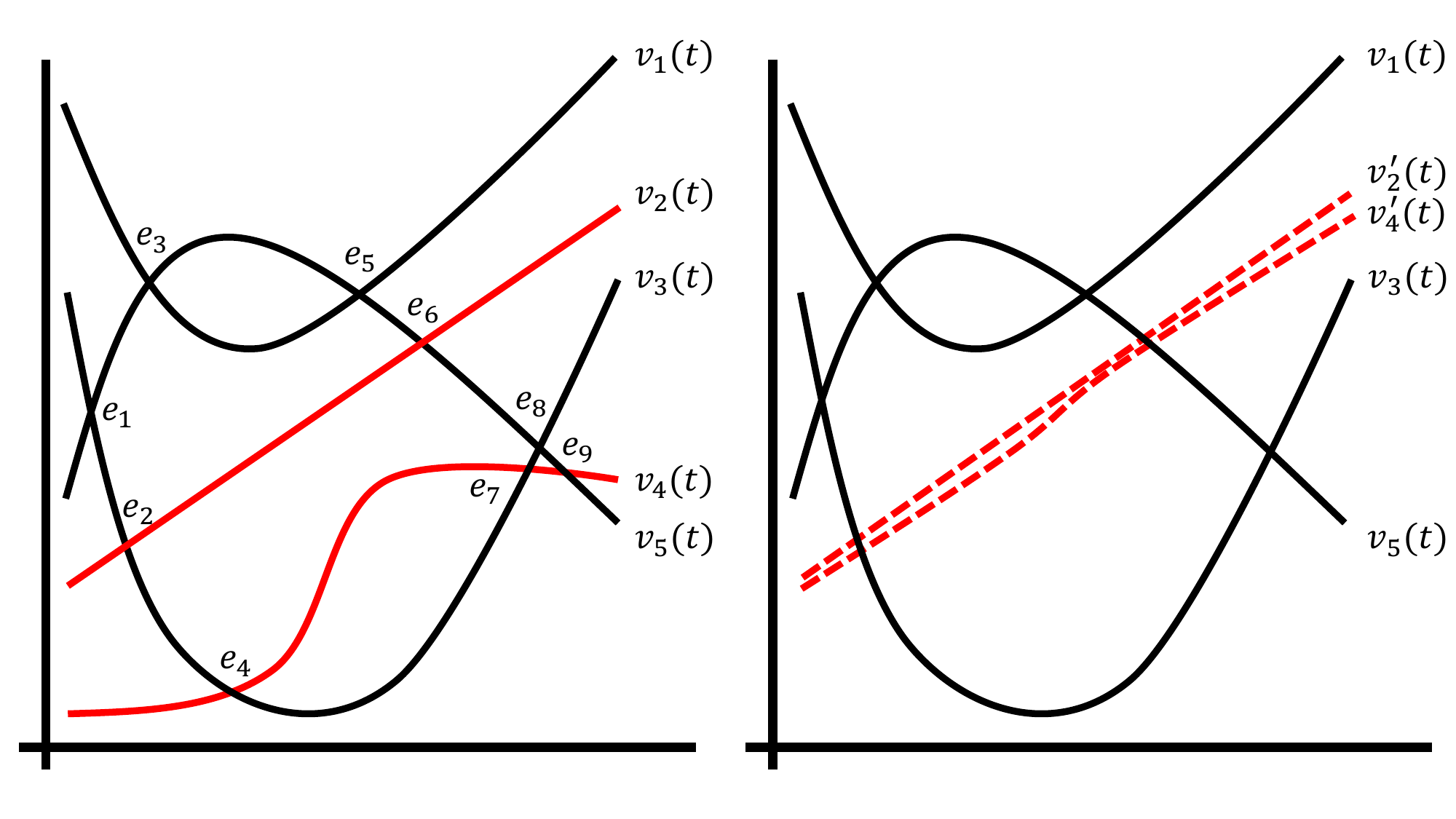}
    \caption{An example of \emph{module shrinking} with five curves and nine collisions: $\mv_{2}(t),\mv_{4}(t)$, colored red, intersect the remaining three curves in the same way. Therefore, we can ``shrink'' the two of them together vertically by an $\epsilon$ factor until they almost become one thin line (dashed), then paste it following any of the original red curves (in this example, $\mv_{2}(t)$). The function graph remains the same after replacing $\{\mv_{2}(t),\mv_{4}(t)\}$ by $\{\mv_{2}'(t),\mv_{4}'(t)\}$.}
    \label{shrink}
\end{figure}
%It is appealing to first identify all modules and take one %representative from each module to study how objects from %different modules intersect each other, and then refine %details within each module. The following proposition shows %that identifying modules is easy.Therefore, we assume there 
%is no module in the rest of the discussion.

It is appealing to first identify all modules and take one 
representative from each module to study how objects from different 
modules intersect, and then subsequently determine patterns within each 
module. Since module decomposition is a well-studied problem in the graph theory 
literature and can be performed in polynomial time~\cite{golumbic2004algorithmic}, 
we hence assume there is no module in function graph in the rest of the discussion.

\begin{lemma}\label{cliq}
For any $S\subseteq\mV$, both $\Max{S},\Min{S}$ are maximal cliques. 
\end{lemma}
\begin{IEEEproof}
We only prove the claim for $\Max{S}$. For any $u,v \in V$ such that $uv \notin E$, we must have $v > u$ or $u > v$. Hence, $u$ and $v$ cannot both be in $\Max{S}$, and $\Max{S}$ must be a clique.
To prove that $\Max{S}$ is a maximal clique, we first note that, for any $w \in S - \Max{S}$, we must have $v \in \Max{S}$ such that $v > w$. Hence, $vw \notin E$, and $w$ cannot be part of a clique with $\Max{S}$.
\end{IEEEproof}

\begin{corollary}\label{laycli}
    Each layer in $\lp\mV_{1},\mV_{2},\hdots,\mV_{k}\rp$ is a maximal clique. The corollary follows from recursively defining  $\mV_{i}:=\Max{(\mV-\cup_{j=1}^{i-1}\mV_{j})}$.
\end{corollary}

\begin{lemma}\label{space}
    Let $\hg=(\mV,\E)$ be a function graph with layers $\lp\mV_{1},\mV_{2},\hdots,\mV_{k}\rp$. Then, for all $\ell\leq i\leq r$ with $ u_{\ell}\in\mV_{\ell},u_{r}\in\mV_{r}\text{,and }u_{\ell}u_{r}\in\E$, there exists a $u_{i}\in\mV_{i}$ with $u_{\ell}u_{i}\in\E$.  
\end{lemma}
\begin{IEEEproof}
We first show that all $u\in\bigcup_{i=\ell}^{r}\mV_{i}$ intersect with at least one of $u_{\ell},u_{r}$.
By definition of layers, there exists a $t^{*}$ such that $u_{\ell}(t^{*})>u(t^{*})>u_{r}(t^{*})$. Suppose $u$ intersects with neither of $u_{\ell},u_{r}$; then, $u_{\ell}(t)>u(t)>u_{r}(t)$ for all $t$, which implies $u_{\ell}>u>u_{r}$ and therefore $u_{\ell}u_{r}\notin\E$, a contradiction. Next, suppose the claim of the lemma is false. 
Then, $\exists\, \ell < i < r$ such that $v_\ell u \notin E$ for all 
$u \in V_i$, which implies $v_r u \in E$ for all $u \in V_i$ by the 
previous discussion. But then $\mV_{i}$ is not a maximal clique among $\bigcup_{j=i}^{r}\mV_{j}$, a contradiction of the definition of layers and Lemma~\ref{cliq}. 
\end{IEEEproof}

\begin{lemma}\label{interval}
    For any function graph $G=(\mV,\E)$ with layers $\lp\mV_{1},\mV_{2},\hdots,\mV_{k}\rp$, let $\hg^{\circ}=(\mV^{\circ},\E^{\circ})$ be its contraction graph with $\mV^{\circ}=\{\mv_{1},\hdots,\mv_{k}\},\E^{\circ}=\{\mv_{i}\mv_{j}|\exists u_{i}\in\mV_{i},u_{j}\in\mV_{j}\text{ with }u_{i}u_{j}\in\E\}$. Then, $\hg^{\circ}$ is an interval graph, i.e., a graph 
admitting a representation in which each vertex corresponds to an 
interval on the real line and two vertices are adjacent if and only if 
their intervals intersect~\cite{fulkerson1965incidence, booth1976testing}.
\end{lemma}

\begin{IEEEproof}
For each $i$, let $i^{*}:=\argmax_{j\geq i}\mathbf{1}\{\mV_{i}\mV_{j}\neq\emptyset\}$. Then by Lemma~\ref{space}, $\mv_{i}\mv_{j}\in\E^{\circ}$ iff $i\leq j\leq i^{*}$. Consequently $\hg^{\circ}=(\mV^{\circ},\E^{\circ})$ is isomorphic to the interval graph with intervals $\{[i,i^{*}]\}_{i=1}^{k}$.
\end{IEEEproof}

The layer decomposition $(V_1, V_2, \ldots, V_k)$ provides the most 
informative positional structure recoverable from unordered collisions: The contraction graph $H^\circ$ directly encodes the relative positions 
of layers, while within each layer, the relative order of objects cannot 
be further distinguished since all objects in the same layer collide with one another. We show next how to efficiently 
recover this decomposition from a given function graph $G$.

%For the rest of this section, we are interested in decomposing function graphs into the form in Lemma~\ref{space}. Because 1) One can easily deduce the relative positions of layers from the contraction graph. 2) Each layer is a set of objects whose relative order cannot be further distinguished from unordered collisions, since all objects within the same layer collide with each other. 

\begin{corollary}\label{mmmc}
For any $S\subseteq\mV$, $\Max{S}\cap\Min{S}=\{u\in S | u\lp S-{u}\rp\subseteq\E\}$.
\end{corollary}
%\begin{IEEEproof}

%Note that $u\in\Max{S}\cap\Min{S}$ implies no $v\in S-{u}$ satisfies $v>u$ or $v<u$, so $uv\in\E$ for all $v\in S-{u}$. Suppose that $u\in S$ connects to all $S-\{u\}$, and $u\notin\Max{S}$ or $u\notin\Min{S}$. Then $u\Max{S}\subseteq\E$ or $u\Min{S}\subseteq\E$. But then $\Max{S}$ or $\Min{S}$ is not maximal, a contradiction to Lemma~\ref{cliq}. 

%\end{IEEEproof}
The proof is omitted. For any $S\subseteq\mV$ it is easy to identify $\Max{S}\cap\Min{S}=\{u\in S | u\lp S-{u}\rp\subseteq\E\}$. Throughout the remainder of the paper, we hence assume $\Max{S}\cap\Min{S}=\emptyset$.
%\textcolor{red}{xxxxxxxxxxxxxxxxxxxxxxxxxxxxxxxxxxxx}

We start with the following definitions.
We say $S\leqs\mV$ if $S\subset\mV$ is ``a lower bound'' on $\mV$, which means that any $\mv_{i}\in S,\mv_{j}\in \mV-S$ with $\mv_{i}\mv_{j}\notin \E$ implies $\mv_{i}<\mv_{j}$. The upper bound is defined in a similar way. For any $u,v,w\in\mV$, we say $\br{u}{v}{w}$ if $uv\in\E$, and $uw,vw\notin\E$. For any $S\subset\mV$, define $\bn{S}:=\{v\in\overline{S}|\exists u\in S, w\in\overline{S}\text{ with }\br{u}{v}{w}\}$ as the bounded neighborhood of $S$. A bounded expansion starting from $S\subset\mV$ is a nested sequence of subsets $\W_{0}(S)\subset\W_{1}(S)\subset\hdots\subset\W_{k}(S)\subseteq\mV$ with $\W_{0}(S)=S,\W_{i}(S)=\W_{i-1}(S)\cup\bn{\W_{i-1}(S)}$, where $\bn{\W_{k}(S)}=\emptyset$. We further define $\W^{*}(S)=\overline{\W_{k}(S)}$ as the set of unreached vertices.
For any $S \subseteq V$, let $G(S)$ denote the induced subgraph of $S$.
\begin{lemma}\label{lbnl}
If $S\leqs\mV$, then $S\cup\bn{S}\leqs\mV$.
\end{lemma}
\begin{IEEEproof}
Suppose the claim is not true. Then, there exists $y\in\bn{S},y^{‘}\in\overline{S\cup\bn{S}}$ such that $y>y^{’}$. By definition of bounded neighbors, there exists $v\in S, x\in\overline{S}$ such that $\br{v}{y}{x}$.
Since $S\leqs\mV$, we have $v<x$, which implies $y<x$, as otherwise $v<x<y$ and therefore $vy\notin\E$, a contradiction. We then have $y^{‘}<y<x$, and therefore $y^{’}x\notin\E$. Moreover $vy^{‘}\notin\E$, otherwise $\br{v}{y^{’}}{x}$ and therefore $y^{‘}\in\bn{S}$, a contradiction.
Then either $v<y^{’}$ or $v>y^{‘}$. Since $S\leqs\mV$, we have $v<y^{’}$. But then $v<y'<y$ and therefore $vy\notin\E$, a contradiction.
\end{IEEEproof}
\begin{lemma}\label{lbnm}
If $S\leqs\mV$, then $\bn{S}\cap\Max{\mV}=\emptyset$.
\end{lemma}
\begin{IEEEproof}
%The proof is similar to that of Lemma~\ref{lbnl}.
Suppose it is not true, then there exists $y\in\bn{S}\cap\Max{\mV}$. By definition of bounded neighbors, there exists $v\in S, x\in\overline{S}$ such that $\br{v}{y}{x}$.
Since $S\leqs\mV$, we have $v<x$, which implies $y<x$, otherwise $v<x<y$ and therefore $vy\notin\E$, a contradiction.
But $y<x$ contradict to $y\in\Max{\mV}$.
\end{IEEEproof}

\begin{lemma}\label{us}
Let $S\subset\mV$ with $C_{1},C_{2},\hdots$ being connected components in the complement graph of $G$, $\overline{G(\overline{S})}$, restricted to $\overline{S}$. Then $\bn{S}=\emptyset$ iff the $C_{i}$'s are modules.
\end{lemma}
\begin{IEEEproof}
$\bn{S}=\emptyset$ implies that for all $v\in S$, and $x,y\in\overline{S}$ with $xy\notin\E$, each $v\in S$ either connect to both $x,y$ or none.
Consequently, for all $C_{i}$, each $v\in S$ either connects to the whole $C_{i}$ or none. Combined with the fact that $C_{i}$ does not connect to any $C_{j},j\neq i$, we have that $C_{i}$ is a module.
\end{IEEEproof}
\begin{lemma}\label{mlemma}
We have $S\leqs\mV$ implying $\W^{*}(S)\subseteq\Max{\mV}$ and $S\geqs\mV$ implying $\W^{*}(S)\subseteq\Min{\mV}$
\end{lemma}
\begin{IEEEproof}
We prove the case of $S\leqs\mV$. Suppose the statement is false, and let $u\in\W^{*}(S)-\Max{\mV}\neq\emptyset$.
By Lemma~\ref{mmmc}, $\Max{\mV}$ is a maximal clique. Therefore, there exists $v\in\Max{\mV}$ with $uv\notin\E$. Otherwise, $\Max{\mV}\cup\{u\}$ is a bigger clique.
Suppose $v\in\overline{\W^{*}(S)}$. Then by Lemma~\ref{lbnl}, $\overline{\W^{*}(S)}\leqs\mV$, which implies $v<u$, a contradiction to $v\in\Max{\mV}$. Therefore, $v\in\W^{*}(S)$ and $u,v\in\W^{*}(S)$ with $uv\notin\E$.
Let $C_{1},C_{2},\hdots$ be connected components in $\overline{G(\overline{S})}$. Then $u,v$ belong to the same component $C_{i}$. By Lemma~\ref{us}, all components $C_{1},C_{2},\hdots$ are modules. Therefore, $C_{i}$ is a module of size $>1$, which contradicts our assumption that $G$ does not have a module with size $>1$.
\end{IEEEproof}
\begin{corollary}\label{c1}
We have $S\leqs\mV$ with $S\cap\Max{\mV}=\emptyset$ implying $\W^{*}(S)=\Max{\mV}$. 
\end{corollary}
\begin{IEEEproof}
By Lemma~\ref{lbnl} and Lemma~\ref{lbnm}, $\overline{\W^{*}(S)}\leqs\mV$ and $\overline{\W^{*}(S)}\cap\Max{\mV}=\emptyset$. Therefore $\Max{\mV}\subseteq\W^{*}(S)$.
By Lemma~\ref{mlemma}, $\W^{*}(S)\subseteq\Max{\mV}$. Therefore $\W^{*}(S)=\Max{\mV}$.
\end{IEEEproof}

\begin{corollary}\label{c2}
If $S\geqs\mV$, then $\W^{*}(S)\leqs\overline{S}$ and $\W^{*}(S)\cap\Max{\overline{S}}=\emptyset$.
\end{corollary}
\begin{IEEEproof}
By Lemma~\ref{mlemma}, $\W^{*}(S)\subseteq\Min{\mV}$. Since $\W^{*}(S)\subseteq\overline{S}$, we have $\W^{*}(S)\subseteq\Min{\mV}\cap\overline{S}\subseteq\Min{\overline{S}}\leqs\overline{S}$.
Since we assumed  $\Min{\overline{S}}\cap\Max{\overline{S}}=\emptyset$, we have $\W^{*}(S)\cap\Max{\overline{S}}=\emptyset$.
\end{IEEEproof}

\begin{theorem}
    Given any lower bound $S\leqs\mV$ with $S\cap\Max{\mV}=\emptyset$, Algorithm~\ref{alg:layers_lb} recovers the layers $\lp\mV_{1},\mV_{2},\hdots,\mV_{k}\rp$.
\end{theorem}
\begin{IEEEproof}
    By Corollary~\ref{c1}, we can recover $\mV_{1}$ from $S$. By Corollary~\ref{c2}, we can recover a new lower bound $S'\leqs \lp\mV-\mV_{1}\rp$ with $S'\cap\mV_{2}=\emptyset$ from $\mV_{1}$. The process proceeds until all layers found.
\end{IEEEproof}
\begin{algorithm}
%\small
\caption{Recover layers given a lower bound}
\label{alg:layers_lb}
\begin{algorithmic}[1]
\STATE \textbf{Input:} A lower bound $S\leqs\mV$ with $S\cap\Max{\mV}=\emptyset$.
\STATE \textbf{Output:} Layers $\lp\mV_{1},\mV_{2},\hdots,\mV_{k}\rp$.
\STATE\COMMENT {Note: $W^*_{G(U)}(A)$ denotes $W^*(A)$ computed on the induced graph $G(U)$.}
\STATE $U\gets V$
\FOR{$i=1,2,3,\hdots$ with $U\neq\emptyset$}
\STATE $V_i\gets W^*_{G(U)}(S)$
\STATE $S\gets W^*_{G(U)}(V_i)$
\STATE $U\gets U\setminus V_i$
%\STATE $\mV_{i} \gets \W^{*}(S)$
%\STATE $S \gets \W^{*}(\mV_{i})$
%\STATE $G \gets G(\mV-\mV_{i})$
\ENDFOR
\end{algorithmic}
\end{algorithm}

In Algorithm~\ref{alg:layers_lb}, we assume a lower bound $S\leqs\mV$ is given. In the rest of the section, we aim to find a lower bound given a function graph $G$.
\begin{lemma}\label{solid}
Let $S_1, S_2 \subseteq V$ be two connected sets such that no edge exists 
between $S_1$ and $S_2$. Then, either the whole $S_{1}$ is ``at the left'' of $S_{2}$ or vise versa, where ``at the left'' means that for all $v_{i}\in S_{1},v_{j}\in S_{2}$, $\mv_{i}\leqs\mv_{j}$.
\end{lemma}
\begin{IEEEproof}
It suffices to show there do not exist $\mv_{1}\in S_{1}$, $\mv_{2},\mv_{2}\in S_{2}$ such that $\mv_{2}<\mv_{1}<\mv_{3}$. Since $S_{2}$ is connected, there exists $\mv_{4}\in S_{2}$ with $\mv_{1}\mv_{4}\in\E$, a contradiction to our assumption that $S_{1},S_{2}$ are independent.
\end{IEEEproof}
\begin{theorem}\label{clibnd}
Let $S\subseteq \mV$ be a maximal clique, $N(S)\subseteq\overline{S}$ the set of neighbors of $S$, and $\{C_{i}\}$ the connected components in $\mV-S-N(S)$. Let $L,R$ be the union of vertices in $C_{i}$'s that are at left and right of $S$, then $L,R\leqs\mV$.
Algorithm~\ref{alg:lb} produces $L,R$.
\end{theorem}
\begin{IEEEproof}
For simplicity, we assume $\mV-S-N(S)\neq\emptyset$.

By Lemma~\ref{solid}, all $C_{i}$s are either at the left or right of $S$. 
\begin{enumerate}
    \item For any $C_{i},C_{j}$ at different sides of $S$, there is no $v\in N(S)$ such that $v$ connect to both $C_{i},C_{j}$. Otherwise, $v$ should connect to all vertices in $S$ and therefore $S$ would not be a maximal clique, a contradiction.
    \item Let $C^{*}\in L$ be leftmost, which means all $C_{i}\in L$ stay between $S,C^{*}$. Since $G$ is connected, there exists $v\in N(S)$ connect to $C^{*}$. Then $v$ also connect to all $C_{i}\in L$. Consequently, all $C_{i},C_{j}$ at same side of $S$ connect to some $v\in N(S)$ simultaneously.
\end{enumerate}
By 1) and 2), we can go through all pairs $C_{i},C_{j}$ and check whether they connect to some $v\in N(S)$ simultaneously, leading us to the partition $L,R$. Also, $L,R\leqs\mV$ follows from their definition.
\end{IEEEproof}
%\vspace{-1mm}
\begin{algorithm}
%\small
\caption{Finding a lower bound}
\label{alg:lb}
\begin{algorithmic}[1]
\STATE \textbf{Input:} A function graph 
$G$.
\STATE \textbf{Output:} A lower bound.
\STATE Find any maximal clique $S$ and the connected components $\mathcal{C}=\{C_1,\ldots,C_m\}$ of $G(V\setminus(S\cup N(S)))$.
\STATE Construct an auxiliary graph $Q$ on $\mathcal{C}$, where $C_iC_j\in E(Q)$ iff some $v\in N(S)$ is adjacent to vertices in both $C_i$ and $C_j$.
\STATE Let $\mathcal{C}_1,\mathcal{C}_2$ be the connected components of $Q$.
\RETURN $\bigcup_{C\in\mathcal{C}_1}C$.
%\STATE Find a maximal clique $S$ \textcolor{blue}{ and  the connected components $\{C_{i}\}$ in $\mV-S-N(S)$.}
%\STATE Go through all pairs $C_{i},C_{j}$ and check whether they simultaneously connect to some $v\in N(S)$.
%\STATE Return either of the two sides as a lower bound.
\end{algorithmic}
\end{algorithm}
%\vspace{-3mm}
\begin{corollary}\label{twostep}
    Given a function graph $G=(\mV,\E)$, we can recover layers $\lp\mV_{1},\mV_{2},\hdots,\mV_{k}\rp$ by first using Algorithm~\ref{alg:lb} to find a lower bound and then applying Algorithm~\ref{alg:layers_lb}.
\end{corollary}
%\vspace{-3mm}
\subsection{Missing Edges}\label{sec:missing}
%\vspace{-2mm}
As shown in Section~\ref{wotwoel}, even if we have no time information in the history, we can still recover layers $V_1,\dots,V_k$ which provide partial position information between the trajectories. We study next a setting with even less information, where some of the collisions are missing. 
Given any graph $G = (V,E)$ with possibly missing collision edges, we 
seek a function graph $G' = (V, E')$ with $E \subseteq E'$ that is 
consistent with all observed collisions while adding as few extra edges 
as possible. Notice that the complete graph is always a valid solution, 
but it is uninformative since every object collides with every other, 
leaving no positional structure to recover. We therefore seek the 
completion $G'$ that preserves as much positional information as 
possible.

To quantify this notion, we observe that in any function graph realization, 
the maximum degree of an object bounds how many others it must collide with 
to realize the given collisions, directly limiting how much positional 
resolution can be extracted. Minimizing the maximum degree of $G'$ thus 
corresponds to finding the least-constrained completion, and we denote 
this minimum achievable maximum degree by $B_f(G)$.
%Let $G^{*}$ be one optimal solution. 
As shown in Lemma~\ref{reduction}, this problem is closely related to the extensively studied graph bandwidth problem~\cite{diaz2002survey,petit2013addenda} that may be succinctly described as follows. Given a graph $G=(V,E)$, find a permutation $\sigma^{*}$ in $\mathbb{S}_n,$ the symmetric group of order $n!,$ that minimizes the worst-case edge stretch, defined as $B(G):=\min_{\sigma}\max_{uv\in E}|\sigma(u)-\sigma(v)|.$ The graph bandwidth problem is NP-hard for general graphs~\cite{diaz2002survey,petit2013addenda}. The best known results for general graphs with efficient placement algorithms ensure an approximation ratio of $O(\log^{3}(n)\sqrt{\log\log(n)})$~\cite{dunagan2001euclidean}. It is NP-hard to approximate the bandwidth of trees within any constant $c \in \mathbb{N},$ even for a subclass of trees called ``caterpillars,'' where every branching vertex (i.e., vertex of degree $>2$) is on the same line~\cite{feige2009approximating}. Moreover, the existence of a $c\sqrt{\frac{\log(n)}{\log\log(n)}}$-approximation, where $c$ is a constant, implies that the NP class admits quasipolynomial time algorithms. %On the other hand, efficient $O\lp\frac{\log(n)}{\log\log(n)}\rp$-approximation algorithms for computing the bandwidth of caterpillars are readily available.
\begin{lemma}\label{reduction}
    Editing $G=(V,E)$ into the lowest maximum degree function graph $G'=(V,E')$ is NP-hard.
\end{lemma}
\begin{IEEEproof}
    %We show $B(G)/2 \leq B_f(G) \leq 2B(G)$ via a straightforward reduction, establishing $B_f(G)$ as a 4-approximation of $B(G)$. Since bandwidth minimization is NP-hard to approximate within any constant $c$~\cite{feige2009approximating}, computing $B_f(G)$ is NP-hard as well. The details are omitted.

\noindent\textbf{Lower bound $B_f(G) \geq B(G)/2$:} Let $G^*$ be an optimal function graph realized by history $H$, and let $p_0$ be the position before the first collision. By definition of $B(G)$, there exists $uv \in E$ such that $u$ and $v$ are at least $B(G)$ apart in $p_0$. Since $u(t)$ and $v(t)$ must collide, every object between them in $p_0$ must collide with at least one of $u, v$. Hence at least one of $u, v$ has degree $\geq B(G)/2$ in $G^*$, giving $B_f(G) \geq B(G)/2$.

\noindent\textbf{Upper bound $B_f(G) \leq 2B(G)$:} Let $\sigma^*$ be an optimal ordering achieving $B(G)$. Initialize object positions according to $\sigma^*$. We then move each objects sequentially starting from the left most one. For each object $u$, let $v$ be its furthest right neighbor in $\sigma^*$; move $u$ rightward past $v$, then return it to its original position, while all other objects remain stationary. The resulting function graph $G'$ satisfies $E \subseteq E'$ and has maximum degree $\leq 2B(G)$, so $B_f(G) \leq 2B(G)$.

Together, these bounds show $B_f(G)$ is a 4-approximation of $B(G)$. Since bandwidth minimization is NP-hard to approximate within any constant $c$~\cite{feige2009approximating}, computing $B_f(G)$ is NP-hard as well.

    %We prove that $B_{f}(G)$ is a $4$-approximation of the bandwidth $B(G)$.
    %Since $G^{*}$ is a function graph, there exists history $\hg=(\mV,\E=\{\e_{i}\})$ realize it.
    %Let $\p_{0}$ be the position before the first collision $\e_{1}$, then by definition of $B(G)$, there exists $u,v\in\mV$ with $uv\in\E$ such that $|\p_{0}(u)-\p_{0}(v)|\geq B(G)$. But $u(t),v(t)$ collide at some point. Therefore all objects in between $u,v$ on $\p_{0}$ must collide to at least one of $u(t),v(t)$, which implies at least one of $u,v$ has degree $\geq B(G)/2$ in $G^{*}$. Consequently, $B_{f}(G)\geq B(G')/2$.   
    
    %On the other hand, let $\sigma^{*}$ be one vertex ordering that achieve $B(G)$. Then the following construction of trajectories from $\sigma^{*}$ result in a function graph $G'=(V,E')$ with $E\subseteq E'$ and maximum degree $\leq 2B(G)$. First, let initial positions of objects follow $\sigma^{*}$. Then we move each objects sequentially starting from the first object in $\sigma^{*}$, and then second one, ... , etc. For each object $u$, let $v$ be the last one in $\sigma^{*}$ with $uv\in E$. We first move $u$ straight to the right until passing $v$, and then move straight to the left until back to original place. During this process, all other objects remain in their position. Let $G'=(V,E')$ be the function graph resulting from this construction of trajectories, then it is not hard to see the maximum degree of $G'\leq 2B(G)$, which implies $B_{f}(G)\leq 2B(G)$.
\end{IEEEproof}
Despite the NP-hardness result, certain structured instances remain 
tractable. 
%We start with the following observation.
%\begin{proposition}Any subgraph of a function graph is still a function graph.\end{proposition}
%At a high level, we aim to edit $G$ into a function graph $G'$ in a recursive fashion. We start with some small subgraphs of $G$, edit them separately into small function graphs, and then merge them into larger ones.
As an example, consider the following simple case: Objects are divided into two groups, $V,W$, with $|V\cap W|=1$. There is no edge between $V,W$.
Moreover, a genie reveals all edges within $V,W$ and therefore $G(V),G(W)$ are function graphs without missing edges. 
Given such a graph $G=(V\cup W,E_{V}+E_{W})$ (two function graphs sharing one vertex), we aim to edit it into a function graph $G'$ with the minimum possible maximum degree.
According to Lemma~\ref{cliq}, function graphs can be shrunk into an interval graph with each vertex being a clique. 
Moreover, one can do so with the two-step framework in Corollary~\ref{twostep}.
Therefore, we are left with interleaving cliques from $G(V),G(W)$ to minimize the worst edge stretch.
For simplicity, we consider the following problem:
%minimizing the maximum number of intersections over layers (cliques).
%In the end, we target at the following
%Note that the existence of $e=vw$ is necessary. Otherwise, putting objects in $V$ at far left and $W$ at far right so that two sets of vertices never intersect is always the best solution.
%Second, instead of minimize the maximum number of intersections over objects, We treat each layer in% $V=(\mV_{1},\hdots,\mV_{k}),W=(W_{1},\hdots,W_{\ell})$ as a giant object and minimize the maximum number of intersections over layers (giant objects). It makes sense to treat each layer as one object, since they are closely connected (they are cliques according to the Lemma~\ref{cliq}).
For a sequence $S=s_{1},s_{2},\hdots,s_{k}$, an interval $I=[s_{i},s_{j}]$ on $S$ has length $j-i+1$, the number of elements between its two endpoints. Given two sequences $X=x_{0},x_{1},x_{2},\hdots,x_{k}$ and $Y=y_{0},y_{1},y_{2},\hdots,y_{\ell}$ with $x_{0}=y_{0}$ and two sets of intervals $I_{X}:=\{[x_{i_{1}},x_{j_{1}}],[x_{i_{2}},x_{j_{2}}],\hdots\},I_{Y}:=\{[y_{i_{1}},y_{j_{1}}],[y_{i_{2}},y_{j_{2}}],\hdots\}$, we wish to interleave $X,Y$ such that the length of the longest interval among $I_{X}\cup I_{Y}$, denoted by $B(X,Y)$,
is minimized.
This is equivalent to saying that for each interval $[x_{i_{1}},x_{j_{1}}]$, one can interleave at most $B(X,Y)-(x_{j_{1}}-x_{i_{1}}+1)$ elements from $Y$ between $x_{i_{1}},x_{j_{1}}$ and similarly for intervals on $Y$.

An interleaving sequence of $X,Y$ is characterized by positions of $X$ relative to $Y$ or vice versa. The position of $x_{i}$ relative to $Y$ is denoted by $\p_{x_{i}|Y}$, where $\p_{x_{i}|Y}=j$ iff $x_{i}$ is in between $y_{j},y_{j+1}$. The position of $X$ relative to $Y$ is denoted by $\p_{X|Y}:=(\p_{x_{1}|Y},\p_{x_{2}|Y},\dots,\p_{x_{k}|Y})$. 

We compute $B(X,Y)$ via binary search on a candidate value $b$, 
checking feasibility by determining 
$p_{X|Y} = (p_{x_1|Y}, \ldots, p_{x_k|Y})$ as follows.
In the \emph{forward pass}, we scan $x_1, \ldots, x_k$ from left 
to right, setting each $p_{x_i|Y}$ to the minimum value 
$\geq p_{x_{i-1}|Y}$ such that all intervals in $I_Y$ containing 
$x_i$ have length $\leq b$, yielding a lower bound $p^1_{X|Y}$.
In the \emph{backward pass}, we scan $x_k, \ldots, x_1$ from right 
to left: for each interval $[x_i, x_j] \in I_X$ whose length exceeds 
$b$ under the current assignment, we update 
$p_{x_i|Y} \leftarrow p_{x_j|Y} - (b-(j-i))$, propagating the 
constraint leftwards and yielding a refined lower bound $p^2_{X|Y}$.
We then \emph{iterate}, alternating between forward and backward 
passes, terminating when either all intervals in $I_X \cup I_Y$ 
have length $\leq b$ (concluding $B(X,Y) \leq b$), or any pass 
forces $p_{x_1|Y} > b$ (concluding $B(X,Y) > b$).
%The whole process looks like a lazy caterpillar moving on the ground. At the first step, the caterpillar tries to move as little as possible. But because there are constraints about the maximum number of bodies each ground can support, each body of the lazy caterpillar has a minimum distance it has no choice but to travel. At the second step, some body parts have to move further since each body part cannot be stretched too long. At the third step, some ground has too many body stand on it again...etc.

%\textcolor{blue}{An alternative way to motivate this problem is by saying we have weaker form of group collision where we know some objects in group $V$ once collision with some objects in group $W$. But we do not know what exactly are these vertices.}

%\textcolor{red}{I still need to filled in some details for reduction at here. But in the end, bandwidth problem will be reduced to the following simpler one.}

\section*{Future Work}
A natural generalization of the setting studied above arises when we drop the assumption that a genie reveals all edges within $V$ and $W$. In that case, the observed two function graphs $G(V)$ and $G(W)$ may themselves have missing edges, so the layer decompositions of $G(V)$ and $G(W)$ --- and hence the orderings of the resulting sequences $X$ and $Y$ --- are no longer uniquely determined. This corresponds precisely to the setting where the internal orderings of $X$ and $Y$ can be freely rearranged, subject only to $x_0 = y_0$ remaining fixed. We call this \textit{interleaving with folding}: one seeks a permutation $\sigma$ of $\{x_0, x_1, \ldots, x_k, y_1, \ldots, y_\ell\}$ with $x_0$ fixed at the leftmost position, minimizing the maximum span of any interval in $I_X \cup I_Y$. Since any order-preserving interleaving is a valid $\sigma$, the optimal fold bandwidth $B_{\mathrm{fold}}(X,Y)$ satisfies $B_{\mathrm{fold}}(X,Y) \leq B(X,Y)$. When only $X$ (or $Y$) is present, folding offers no advantage: the original ordering is already optimal. When both $X$ and $Y$ are present, however, the interaction between them is possible to make folding strictly beneficial. Consider the special case where $X$ and $Y$ are identical, with $I_X = I_Y$. The uniform order-preserving interleaving $(x_0, x_1, y_1, x_2, y_2, \ldots)$ aligns every long interval of $X$ with the corresponding long interval of $Y$, doubling the bandwidth, which is the worst case. A nontrivial folding of $X$ and $Y$ can break this alignment and potentially achieve strictly lower bandwidth than interleaving $X$ and $Y$ without folding. Characterizing $B_{\mathrm{fold}}(X,Y)$ and designing an efficient algorithm to compute it remain open problems.